\documentclass[11pt]{amsart}
\usepackage{amsxtra,bbold}
\usepackage{amssymb}
\usepackage{graphicx}
\theoremstyle{plain}

\newmuskip\pFqmuskip

\newcommand*\pFq[6][8]{%
	\begingroup 
	\pFqmuskip=#1mu\relax
	\mathcode`\,=\string"8000
	\begingroup\lccode`\~=`\,
	\lowercase{\endgroup\let~}\pFqcomma
	{}_{#2}F_{#3}{\left[\genfrac..{0pt}{}{#4}{#5};#6\right]}%
	\endgroup
}
\newcommand{\pFqcomma}{\mskip\pFqmuskip}

\newtheorem{theorem}{Theorem}[section]
\newtheorem{lemma}[theorem]{Lemma}
\newtheorem{proposition}[theorem]{Proposition}
\newtheorem{corollary}[theorem]{Corollary}
\newtheorem{definition}[theorem]{Definition}
\newtheorem{example}[theorem]{Example}

\newtheorem{remark}[theorem]{Remark}
\newtheorem{conjecture}[theorem]{Conjecture}
\newtheorem{problem}{Problem}
\newtheorem{question}[theorem]{Question}
\newcommand \bth[1] { \begin{theorem}\label{t#1} }
	\newcommand \ble[1] { \begin{lemma}\label{l#1} }
		\newcommand \bpr[1] { \begin{proposition}\label{p#1} }
			\newcommand \bco[1] { \begin{corollary}\label{c#1} }
		\newcommand \bde[1] { \begin{definition}\label{d#1}\rm }
					\newcommand \bex[1] { \begin{example}\label{e#1}\rm }
						\newcommand \bre[1] { \begin{remark}\label{r#1}\rm }
							\newcommand \bcon[1] { \begin{conjecture}\label{con#1}\rm }
								\newcommand \bque[1] { \begin{question}\label{que#1}\rm }
			\newcommand{\beq}{\begin{equation}}
			\newcommand{\eeq}{\end{equation}}
			\newcommand{\beqa}{\begin{eqnarray}}
			\newcommand{\eeqa}{\end{eqnarray}}
			\newcommand{\beaa}{\begin{eqnarray*}}
				\newcommand{\eaa}{\end{eqnarray*}}
				\newcommand {\ethe} { \end{theorem} }
		
								\newcommand {\ele} { \end{lemma} }
		\newcommand {\epr} { \end{proposition} }
						\newcommand {\eco} { \end{corollary} }
					\newcommand {\ede} { \end{definition} }
				\newcommand {\eex} { \end{example} }
			\newcommand {\ere} { \end{remark} }
		\newcommand {\econ} { \end{conjecture} }
	\newcommand {\eque} { \end{question} }

\newtheorem{THEO}{Theorem}

\newtheorem{PROB}{Problem}

\def \Cset {{\mathbb C}}

\def \ad { {\mathrm{ad}} }

\newcommand{\al}{\alpha}
\newcommand{\be}{\beta}

\newcommand{\bC}{\mathbb C}
\newcommand{\bR}{\mathbb R}

\def \Cset {{\mathbb C}}

\def \de {\delta}
\def \al {\alpha}
\def \be {\beta}

\def \la {\lambda}
\def \La {\Lambda}

\def \ga {\gamma}
\def \de {\delta}

\usepackage{hyperref}

\makeindex
\begin{document}
\title[In search of a higher Bochner theorem]
 {In search of a higher Bochner theorem}
 		\author[E.~Horozov]{Emil Horozov}
 	\address{
 		Department of Mathematics and Informatics \\
 		Sofia University \\
 		5 J. Bourchier Blvd. \\
 		Sofia 1126 \\
 		Bulgaria, and\\	
 		Institute of Mathematics and Informatics, \\
 		Bulg. Acad. of Sci., Acad. G. Bonchev Str., Block 8, 1113 Sofia,
 		Bulgaria	}
 	\email{horozov@fmi.uni-sofia.bg}

\author[B.~Shapiro]{Boris Shapiro}

\address{
Department of Mathematics \\
Stockholm University, S-10691, Stockholm, Sweden
}
\email{shapiro@math.su.se}

\author[M.~Tater]{Milo\v{s} Tater}
\address{Department of Theoretical Physics, Nuclear Physics Institute,
Academy of Sciences, 250\,68 \v{R}e\v{z} near Prague, Czech
Republic}
\email{tater@ujf.cas.cz}



\dedicatory{To Salomon Bochner, a mathematical hero}

\date{\today}

\date{}

\begin{abstract} We initiate the study of  a natural generalisation of the classical Bochner-Krall problem asking  which  linear ordinary differential operators possess sequences of eigenpolynomials satisfying  linear recurrence relations of finite length; the classical case corresponds to the $3$-term recurrence relations with real coefficients subject to some extra restrictions.  We formulate a general conjecture and prove it in the first non-trivial case of operators of order $3$.
\end{abstract}

\maketitle


\medskip

\section{Introduction}\label{intro}

\subsection{On the classical problem}
In 1929 S.~Bochner published a short paper \cite{Bo} dealing with orthogonal polynomials and Sturm-Liouville problem.\footnote {Salomon Bochner made substantial contributions to harmonic analysis, probability theory, differential geometry  as well as history of mathematics. Several notions and results such as  the Bochner integral, Bochner theorem on Fourier transforms, Bochner-Riesz means, Bochner-Martinelli formula bear at present his name.  His mathematical production consists of 6 books on various subjects including history of mathematics and about 140 research papers. Among these papers 32 were published in Proc. Nat. Acad. USA, 46 in Annals of Mathematics, 3 in Acta Math., and 4 in Duke Math. J.  
He belonged to a  sizeable group of European mathematicians of Jewish origin who moved to US before or during the WWII  and   contributed to an enormous development of  mathematics in their new motherland.}
Although  after writing \cite {Bo}, he left this area  for good, his importance for the theory of orthogonal polynomials is difficult to overestimate;  at the moment \cite {Bo} has been cited 453 times.\footnote{As was pointed out by an anonymous referee the main result of \cite{Bo} is essentially contained in an earlier paper \cite{Ro} with a very fuzzy title.} 



\smallskip
Namely, the following classification problem  was formulated by S.~Bochner for differential operators of order $2$, and by H.~L.~Krall for differential operators of general order.

\begin{PROB}[\cite{Bo, Kr_{21}}]\label{BK-problem}
Describe all linear differential operators  with  polynomial coefficients of the form:
\begin{equation}\label{opBK}
L=L(x,\partial)=\sum_{i=1}^k a_i(x)\partial^i,
\end{equation}
such that a)  $\deg a_i(x) \le i$; b)  there exists a positive integer  $i_0\le k$ with $\deg a_{i_0}(x)=i_0$,
satisfying the condition that the set of polynomial solutions of the formal spectral problem
\begin{equation}\label{eq:specprob}
L f(x)=\lambda f(x),\quad  \lambda\in\bR,
\end{equation}
 is a sequence of polynomials orthogonal with respect to some real bilinear form.  (Here $\partial:=\frac{d}{dx}$.)
\end{PROB}



Following the terminology used in physics, we call linear differential operators given by \eqref{opBK}  {\it exactly solvable}, see e.g. \cite{Tu, RuTu}.  Observe that every exactly solvable operator has a unique eigenpolynomial of any sufficiently large degree which makes Problem~\ref{BK-problem} well-posed.

Let us denote by $\{P_n^L(x)\}$  the sequence of monic eigenpolynomials of an exactly solvable operator $L$.  (We assume that $\deg P_n^L(x)=n$ where $n$ runs from some positive integer to $+\infty$.) An exactly solvable operator which solves Problem~\ref{BK-problem}   will  be  called  a   {\it Bochner-Krall operator}. 

\medskip
  S.~Bochner stated and solved Problem~\ref{BK-problem} for the second order differential operators, see  \cite{Bo} and Theorem~\ref{th:Boch} below. The second  order classification contains four families, corresponding to the classical Hermite, Laguerre, Jacobi, and Bessel polynomials. Eleven years after that H.~L.~Krall settled the order four case, see \cite{Kr2}. The order four classification contains seven families: the four classical families corresponding to the case when an order four operator is a function of  an operator of order two, and three new families which are polynomial eigenfunctions of differential operators that do not reduce to operators of order two, see \cite{Kr_{21}} and
 Theorem~\ref{th:Krall} below.

 An assortment of families corresponding to order six operators has been found in the ensuing decades, see e.g. \cite[Chapter XVI]{Krall-junior}.    W.~Hahn showed that the four classical families are the only orthogonal polynomial sequences $\{P_n(x)\}_{n=0}^\infty$ for which $\{\tfrac{d}{dz}P_n(x)\}_{n=0}^\infty$ is also an orthogonal polynomial sequence, see \cite{Hahn}. An analog of this was proved for the order 4 case by K.~H.~Kwon, L.~L.~Littlejohn, J.~K.~Lee, and B.~H.~Yoo \cite{KLLY}.  The most general form of Problem \ref{BK-problem}
is still open for operators of order six or higher, but K.~H.~Kwon and J.~K.~Lee have found a  satisfactory solution if the polynomials are required to be orthogonal with respect to a compactly supported positive measure, see \cite{KL}. (A weaker result in the same direction was somewhat earlier obtained  in \cite{BRSh}.)

\medskip
In Problem \ref{BK-problem}, one may equivalently seek a sequence of moments $\{\mu_j\}_{j=0}^\infty$ which permits construction of the orthogonal sequence of polynomials, see \cite[p. 223]{Krall-junior}.  Let $\langle \cdot, \cdot \rangle$ be a candidate bilinear form, and define a weight on the set of polynomials with respect to the inner product by requiring that $\langle w(x), x^n\rangle = \mu_n$. (Note that such a weight $w(x)$ is not necessarily positive or unique.)  One can show that solutions to Problem \ref{BK-problem} exist only if the product $wL$ is equal to its formal adjoint when acting on polynomials, and this immediately implies that the order of $L$ must be even, see  \cite{Krall-junior}, p. 228.  

\medskip
Recall that  by a trivial generalisation of Favard's theorem, every  sequence of monic orthogonal polynomials satisfies a $3$-term recurrence relation of the form:
\begin{equation}\label{3-term}
xP_n(x) = P_{n+1}(x) +u(n) P_n(x) + v(n) P_{n-1}(x),
\end{equation}
where $u(n)$ and $v(n)$ are real numbers.
 Observe that the weight function $w(x)$ is non-negative in $\bR$ if and only if $v(n)>0$ for all $n >0$.

\subsection{Algebraisation and generalisation}

\medskip
The following natural algebraic version of the classical Bochner-Krall problem  was studied in substantial details in \cite{GHH, GY}.  

\begin{problem}[algebraic Bochner-Krall problem] \label{CBK-problem}
Describe all linear differential operators  \eqref{opBK}
such that their sequence of monic eigenpolynomials satisfies \eqref{3-term} with  some complex-valued $u(n)$ and $v(n)$.
\end{problem}

 \medskip
Observe that contrary to the case of the classical Bochner-Krall problem, Problem~\ref{CBK-problem} is  purely algebraic and (hopefully) has an easier solution compared to the classical one, see Conjecture~\ref{con1} below. On the other hand, the connection between the algebraic Bochner-Krall problem and the classical one is rather straight-forward.

\begin{proposition} \label{prop:translate} A differential operator solves the classical Bochner-Krall problem with a positive weight function $w(x)$ if and only if it solves Problem~\ref{CBK-problem} and all its eigenpolynomials are real-rooted and the roots of any two consecutive eigenpolynomials are strictly interlacing.
\end{proposition}

\begin{remark}{\rm Which additional conditions on an exactly solvable operator can guarantee the real-rootedness of its eigenpolynomials is unclear at the moment, but this could be related to the major results in the P\'olya-Schur theory, see e.g. \cite{BB}.
}
\end{remark}

\begin{remark}{\rm Observe that if we allow coefficients of the operator to depend on $n,$ then the whole problem trivializes, i.e., more or less any sequence of polynomials can be obtained in such a way.
}
\end{remark}

\medskip
 In what follows, we  also discuss the following natural generalisation of Problem~\ref{CBK-problem}.

\begin{problem}[generalised algebraic Bochner-Krall problem]\label{prob:genBK} Describe the set of linear differential operators  \eqref{opBK} such that their sequence of monic eigenpolynomials satisfies a finite recurrence relation of the form:
\begin{equation}\label{eq:MOP}
xP_n(x) =   P_{n+1}(x)  + \sum_{j=0}^{d} b_j(n)P_{n-j}(x)=\La P_n(x),\; n=0,1,\dots 
\end{equation}
	where $d$ is independent of $n$ and the coefficients $b_j(n)$ are independent of $x$. Here $\La$ is a difference operator acting on infinite column vectors $(\mu_0, \mu_1,\dots)^T$; interpreting $\La$ as an infinite matrix we get that its $n$-th row corresponds to   $$(\La)_n=T+\sum_{j=0}^db_j(n)T^{-j}$$ and $T(\mu_0, \mu_1,\dots) :=(\mu_1,\mu_2, \dots)$ is the shift operator.
\end{problem}

\medskip
\begin{remark} {\rm By a theorem of Maroni   \cite{Ma}, the latter condition   means that  the sequence $\{P_n(x)\}$ of polynomials is $d$-orthogonal, see e.g., \cite{VI}.  The study of $d$-orthogonal and multiple orthogonal polynomials has been a popular area of research during the last 3 decades, see e.g., \cite{ApKu}.
}

\end{remark}

\begin{remark}{\rm
 It is clear that if an operator $L$ solves Problems~\ref{CBK-problem} or ~\ref{prob:genBK}, then for   any univariate polynomial $s$ of positive degree,  the operator  $s(L)$ also solves the same problem with the same sequence of eigenpolynomials $\{P_n^L(x)\}$ and the sequence of eigenvalues coinciding with $\{s(\la_n)\}$. Therefore in what follows, we will always assume that $L$  is the  operator of the lowest order with a given sequence of eigenpolynomials. }
\end{remark}

\begin{remark}{\rm
Another way to produce new solutions to Problems~\ref{CBK-problem} and ~\ref{prob:genBK} from the existing ones is  to apply a bispectral Darboux transformation to the difference operator $\La$. This will produce a new difference operator $\widehat\La$ and a new differential operator $\widehat L$. However the order of $\widehat L$ could increase and, in fact, this  is what happens in all known examples.
We will call a differential operator $L$  \emph {irreducible} if it  cannot be obtained by applying a bispectral Darboux transformation to a differential operator of a  lower order. (For general information about Darboux transformation consult e.g. \cite{BrBr}). 
}
\end{remark}

Most of the above situations allow application of Darboux transformations on the side of the difference operator. More precisely, let us present the operator $\La$ as a matrix, acting on the column vector $\mathcal P(x)=(P_0(x), P_1(x), \dots)^T$ of polynomials in the variable $x$. Then in certain situations there exists a factorisation of $\La$ 
$$\La = D_1 \circ D_2,$$
where $D_1$ and $D_2$  are rational functions of $n$, and, additionally,  $D_1$ is an upper triangular matrix while $D_2$ is a lower triangular 
matrix. Set
$$\widehat \La = D_2 \circ D_1 $$
and define a new column vector $\widehat {\mathcal P}(x)=(\widehat P_0(x), \widehat P_1(x), \dots)^T$ of polynomials by using
$$ \widehat {\mathcal{ P}}(x) = D_2{\mathcal P}(x).$$
It is easy to see that $\deg \widehat P_n(x) = n$ and also that 
$$x \widehat {\mathcal {P} }( x ) = \widehat \La \widehat {\mathcal P}( x ).$$
In such a situation a general procedure described in \cite{BHY1} allows us to obtain a univariate differential operator $\widehat L$, such that
$$\widehat L\widehat P_n(x) = \nu_n\widehat P_n(x), \;n=0,1,\dots.$$

\medskip
To carry out  this procedure  in concrete cases is a nontrivial problem. However  in \cite{GHH,  GY} it has been successfully performed in the situations when one starts either from the Laguerre or from  the Jacobi polynomials. In particular, using a  $1$-step Darboux transformation one obtains Krall's polynomials. Analogously, many-step Darboux transformations lead to differential operators of higher order, but they do not increase the order of the difference operator.

\medskip
In this language,
 the above mentioned result of  S.~Bochner \cite{Bo} can be stated  as follows.

\begin{THEO}\label{th:Boch} Every orthogonal polynomial system $\{P_n(x)\}$ obtained as  a sequence of monic eigenpolynomials of a linear differential operator of order $2$ coincides (up to the action of the affine group on the variable $x$ and scaling) with one of  the sequences of  Hermite, Laguerre, Jacobi or Bessel orthogonal polynomials, i.e.,  with one of the classical orthogonal polynomial systems. \end{THEO}

Similarly,  the above mentioned result of  H.~Krall claims the following, see \cite{Kr2}.

\begin{THEO}\label{th:Krall}
Every linear differential operator of order $4$ solving the classical Bochner-Krall problem and which can not be obtained as a function of a linear differential operator of order $2$ is given  either by

\medskip
\noindent
\rm{(a)}   a 1-step Darboux transformation applied to  an  operator of order $2$ corresponding to   a system of Jacobi polynomials with at least one of the parameters $\al, \be$ being an  integer,

\medskip
\noindent
or by

\medskip
\noindent
\rm{(b)}   a 1-step Darboux transformation applied to an operator corresponding to a system of  Laguerre polynomials with an integer value of its parameter $\al$.
\end{THEO}

\medskip Let us now present our  conjectures and results related to Problems~\ref{CBK-problem} and ~\ref{prob:genBK}.
Concerning Problem~\ref{CBK-problem},   we propose the following general guess supported by  the results of  \cite{GHH, GY}.\footnote{As was pointed out by an anonymous referee Conjecture~\ref{con1} can be found on p. 153 of \cite{GrH}.} 

 \begin{conjecture}\label{con1}
 	Any sequence $\{P_n\}$ of  monic polynomials obtained as a sequence of eigenpolynomials of a linear differential  operator solving Problem~\ref{CBK-problem} 
	belongs to one of the following 2 classes:
 	
 	\bigskip

 \noindent	
 	{\rm 1)} A classical sequence of orthogonal polynomials with, in general, complex-valued parameters, i.e.,
 	
 	\smallskip
 	
	\noindent
 	{\rm (a)} Hermite polynomials $H_n$;
 	
 		\smallskip
 	
	\noindent
 	{\rm (b)} Laguerre polynomials $L^{(\alpha)}_n;$
 	
 		\smallskip
 	
	\noindent	
 	{\rm (c)} Jacobi polynomials $P^{(\alpha, \beta)}_n;$
 	
 	 		\smallskip

    \noindent
    {\rm (d)} Bessel polynomials $Y_n$.
 	
 	\bigskip
 	
	\noindent
 	{\rm 2)} A sequence of polynomials which can be obtained  from

 	\smallskip

 \noindent
 	{\rm (e)}  Laguerre polynomials $L^{(\alpha)}_n$ with the positive integer values of  $\alpha$ by applying a  finite number of Darboux transformations to the difference operator $\La$ with $$(\La)_n:=T + u(n)Id  + v(n)T^{-1}$$ occurring
	 in the right-hand side of \eqref{3-term};
 	
 		\smallskip

 \noindent	
 	{\rm 	(f)}  Jacobi polynomials $P^{(\alpha, \beta)}_n$  with the positive integer values of either $\alpha$, $\beta$ or both  $\alpha$ and  $\beta$ by applying a finite number of Darboux transformations to the difference operator $\La$.
 \end{conjecture}

\begin{remark} {\rm The fact that the families (a) - (d) solve Problem~\ref{CBK-problem} is classical. Analogous statement for the families (e)-(f) has been proven in \cite{GHH} and \cite{GY}.   The major difficulty in settling Conjecture~\ref{con1} is to show that the above list of cases is exhaustive which is possible in principle, but computationally is quite a challenge. }
\end{remark}

Concerning Problem~\ref{prob:genBK},    we have the following two conjectural claims. %

\begin{conjecture}\label{conj:order}
For any irreducible differential operator $L$ of order $k$ solving Problem~\ref{prob:genBK}, the order of the corresponding difference operator $\La$ is also $k$.
\end{conjecture}

The next claim similar to Conjecture~\ref{con1} gives   a  description of all irreducible operators solving Problem~\ref{prob:genBK}.

\begin{conjecture}\label{conj:genBK} For any positive integer $k$, the irreducible differential operators of order $k$ solving Problem~\ref{prob:genBK} belong to one of the following  two types:

\smallskip
\noindent
$${ \rm 1)} \hskip2cm  \quad L= \sum_{j=1}^ka_jx^{j-1}\partial^j +x\partial,\quad a_j \in \bC,\; a_k\neq 0,$$
generating the so-called $(k-1)$-orthogonal polynomials, see \cite{BChD};

\smallskip
\noindent
$${ \rm 2)} \hskip2cm  L = q^\prime(G)G + x\partial, \quad \quad \quad \quad \quad \quad \quad \quad \quad \quad \quad$$
where  $q(t)$ is any complex polynomial of degree $k/\ell$ without a constant term  with $\ell$ being any divisor of $k$. Additionally, $G := ( \sum_{m=0}^{\ell-1} a_m(x\partial)^m)\partial,\quad  a_m \in \bC, a_{\ell-1}\neq 0,$ see Theorem 2.3 of \cite{Ho1}. 
\end{conjecture}

\begin{remark} {\rm Both types of operators appearing in Conjecture~\ref{conj:genBK} together with their properties were discussed  in some detail in  publications  \cite{Ho,Ho1}.}
\end{remark}

\begin{remark}{\rm
Notice that the operators
$$ L= \sum_{j=1}^ka_j\partial^j +x\partial$$  generating the Appell polynomials \cite{App}
are included in Type 2  with $\ell = 1$, see \cite{Dou, Ho1}. }
\end{remark}

\medskip
\begin{example} Conjecturally, all the irreducible  operators of order 4 solving Problem~\ref{prob:genBK} are given by:

 \noindent
$${ \rm 1)} \hskip2cm L= \sum_{j=1}^4a_jx^{j-1}\partial^j +x\partial,\quad a_j \in \bC, a_4\neq 0$$
generating the so-called $3$-orthogonal polynomials;

\smallskip
\noindent
$$ {\rm 2^\prime)}  \hskip2cm L= \sum_{j=1}^4 a_j\partial^j +x\partial,\quad a_j \in \bC, a_4\neq 0 \quad \quad$$
generating the Appell polynomials, see \cite{App, GH};

\smallskip
\noindent
$${\rm 2^{\prime\prime})}  \hskip2cm L=b_2G^2 +b_1G +x\partial, \quad b_2\neq 0,\quad \quad\quad \quad$$
 where $G = (a_1x\partial+a_0)\partial,\; a_1\neq 0$.

\smallskip
 The difference operator $\La$ in all these cases is of order $4$, i.e.,  $d=3$.

\end{example}

\medskip
The main result of the present paper is a proof of Conjecture~\ref{conj:genBK} in the first non-trivial case of differential operators $L$ of order $3$.

 \begin{theorem}\label{th:main} All irreducible differential  operators $L$ of order 3 solving Problem~\ref{prob:genBK} whose corresponding difference operators $\La$  also have order $3$  are given by:

\noindent
  $${\rm 1)} \hskip2cm L= \sum_{j=1}^3a_jx^{j-1}\partial^j +x\partial,\quad  a_j \in \bC, a_3\neq 0,$$
generating the  $2$-orthogonal monic polynomials, see \cite{BChD}. They satisfy the $4$-term recurrence relation:
$$xP_n(x)=P_{n+1}(x) -(a_1+2na_2+3n(n-1)a_3)P_n(x)+n(a_2+(3n-3)a_3)$$  $$(a_1+(n-1)a_2+
(n-1)(n-2)a_3)P_{n-1}(x)-n(n-1)a_3(a_1+(n-1)a_2+(n-1)(n-2)a_3)$$ $$(a_1+(n-1)a_2+(n-1)(n-2)a_3)P_{n-2}(x),$$
with the standard initial conditions: $P_{-2}(x)=P_{-1}(x)=0, P_0(x)=1$;

\smallskip
 \noindent
 $${\rm 2) } \hskip2cm L= \sum_{j=1}^3 a_j\partial^j +x\partial,\quad a_j \in \bC, a_3\neq 0 $$
generating the monic Appell polynomials, see \cite{App, GH}.
They satisfy the $4$-term recurrence relation:
$$xP_n(x)=P_{n+1}(x)-a_1P_n(x)-a_2nP_{n-1}(x)-a_3 n (n-1)P_{n-2}(x),$$
with the standard initial conditions: $P_{-2}(x)=P_{-1}(x)=0, P_0(x)=1$.

 \end{theorem}









\begin{remark} {\rm Theorem~\ref{th:main} will follow from the detailed study of three special cases presented in \S~3, \S~4, and \S~5 respectively. Its proof  contains several new ideas, but also a substantial amount of explicit calculations which we were able to carry out by hands in \S~3 and \S~5. In \S~4 we had  to use Wolfram Mathematica package since these calculations were too heavy. It is of course highly desirable to find an alternative proof avoiding such an amount of explicit calculations and which might be applicable for a possible attack on our general Conjecture~\ref{conj:genBK}, but at the moment we do not see such a possibility.
}
\end{remark}

\medskip
\noindent
{\it Acknowledgements.} The first and the third authors are sincerely grateful to the Mathematics Department   of Stockholm University for the hospitality in December 2014 and April 2015, when this project was initiated. The first and the second authors want to thank the Department of Theoretical Physics, Nuclear Physics Institute,
Academy of Sciences, 250\,68 \v{R}e\v{z} near Prague for the hospitality in December 2016.    The research of the first author has been partially supported by Grant KP-06-N 62/5
of Bulgarian National Science Fund.   The third author  has been supported by the Czech Science
Foundation (GA\v{C}R) within the project 21-07129S. Finally, the authors are sincerely grateful to both anonymous referees for their constructive criticism and detailed suggestions which helped us to improve the exposition.  


 \section{Preliminary facts}
 
 Let us write the three-term recurrence relation for monic polynomials in the form
 
 \begin{equation}\label{eq:3trm}
 xP_n =P_{n+1} + u_nP_n +v_{n-1}P_{n-1}.
 \end{equation}

 \begin{proof}[Proof of Proposition~\ref{prop:translate}] It is a well-known fact  that  polynomials in a sequence $\{P_n\}$ of (monic orthogonal) polynomials which satisfy a $3$-term recurrence relation \eqref{eq:3trm} with real $\{u_n\}$ and positive $\{v_n\}$ are all real-rooted and having simple interlacing zeros. As a special case the same property holds for the sequences of eigenpolynomials solving the classical Bochner-Krall problem.  The converse claim that a sequence of monic polynomials satisfying \eqref{eq:3trm}  with real $\{u_n\}$ and $\{v_n\}$ which have all real and interlacing roots necessarily has 
all $v_n$ positive is an immediate consequence of  Wendroff's theorem which, in its turn, is a simple consequence of Sturm's algorithm and should be better known, see \cite {We}. 
\end{proof} 
 
 
 
 

 



 \
Problem~\ref{prob:genBK} asks to describe  all exactly solvable linear differential operators   $L$  that have a sequence $\mathcal P(x)=\{P_n(x)\}$  of monic eigenpolynomials with eigenvalues $\la_n$, i.e.,

 \begin{equation} \label{Bis1}
 LP_n(x) = \la_n P_n(x)
 \end{equation}
 and which at the same time satisfy a difference equation of the form

  \begin{equation} \label{Bis2}
 \La \mathcal P(x) = x \mathcal P (x).
 \end{equation}

 In the next lemma we deduce a simple, but very useful necessary condition for the solvability of the latter problem for differential operators of order $3$. Let us assume that a polynomial sequence  $\mathcal P(x)=\{P_n(x)\}$ satisfying \eqref{Bis1} and \eqref{Bis2} exists. Define
 $$\ad_x L := [x, L]$$ \quad \text{and} \quad  $$ \ad^{j+1}_x L := \ad_x (\ad^j_x L).$$
  It is obvious that $\ad_x L$ is a univariate differential operator of order less than the order of $L$. Below we assume that  the order  of  $L$ equals 3, i.e.,

  \begin{equation}\label{eq:L}
  L = a_3(x)\partial^3 + a_2(x)\partial^2 + a_1(x)\partial.
  \end{equation}

  Let us deduce  that
 \begin{equation}\label{eq:j4}
 \ad^4_x L \equiv 0.
 \end{equation}
 
Indeed, 
$$\ad_x L(u(x))=x L(u(x))-L(x u(x))=-2 a_2(x) u'(x)-3 a_3(x) u''(x)-a_1(x) u(x).$$

Similarly,
$$\ad^2_x L(u(x))=L\left(x^2 u(x)\right)+x (x L(u(x))-L(x u(x)))-x L(x u(x))=6 a_3(x) u'(x)+2 a_2(x) u(x).$$

Finally,
\begin{equation}\label{eq:imp}
\ad^3_x L(u(x))=-6 a_3(x) u(x) \quad \text{and} \quad \ad^4_x L(u(x))\equiv 0.
\end{equation}

 \medskip
 The following result obtained first in \cite{DG} will be crucial in all our calculations. We state it in a form suitable for the present paper. For a column vector $\mathcal {M}=(\mu_0,\mu_1,\dots)^T$, in what follows we use the notation 
 $$\ad^{j+1}_\La\mathcal M :=\ad_{\La}(\ad^j_{\La}\mathcal M) $$ with $\ad_{\La}\mathcal M$ meaning $[\La, \mathcal M]$, i.e., the commutator of the operator $\La$ and   the operator $\mathcal M \cdot \mathbb{1}$ acting diagonally on a sequence of numbers by multiplication of  its $j$-th entry by $\mu_j,\; j=0,1,\dots$. (However slightly abusing our notation  we will omit the symbol $\mathbb{1}$ below).   

 \ble{ad}\label{lm2.1} Assuming that the pair $(L,\Lambda)$ solves Problem 2 with $L$ being a differential operator of order $3$  and 
 denoting by $\mathcal L =(\la_0, \la_1,\dots)^T$, the sequence of its eigenvalues, one has the relation
 \beq \label{ad}
 \ad^4_{\La}\mathcal L \equiv 0,
 \eeq
 where the left-hand side is understood as a difference operator acting on column vectors. 
 \ele
 

 \proof  
 
 
 

By definition, $[x, L](P_n(x)) = x L(P_n(x)) -  L (xP_n(x))$. Introducing the column vector $\mathcal P=(P_0(x),P_1(x),\dots )^T$ of monic eigenpolynomials of $L$, we have the equality of vectors of functions 
 \[ x L(\mathcal P(x)) -  L (x\mathcal P)
=  x \mathcal L( \mathcal P(x)) -  L (\La (\mathcal P(x))), 
\]
which follows from \eqref{Bis2}. Here (as above) the expression $\mathcal L (\mathcal P(x))$ means that each monic eigenpolynomial $P_j(x),\; j=0,1,\dots$ is multiplied by its own eigenvalue $\la_j$. 
Using the fact that both $\mathcal L$ and $\La$ are independent of $x,$ we can  rewrite the above as
\[
x \mathcal L (\mathcal P(x)) -  L (\La (\mathcal P(x))) =   \mathcal L (x\mathcal P(x)) -    \La (L (\mathcal P(x)))=\mathcal L (\Lambda(\mathcal P(x)))- \La (\mathcal L (\mathcal P(x))).
\]
Summarizing, we obtain
\[
[x, L](\mathcal P(x)) = \mathcal L (\La (\mathcal P(x))) -  \La (\mathcal L (\mathcal P(x)))  =  [\mathcal L, \La](\mathcal P(x)),
\]
which is equivalent to  
\begin{equation}\label{eq:impeq}
\ad_{\La}\mathcal L(\mathcal P(x)) = - \ad_{x}L(\mathcal P(x)). 
\end{equation}

Observe that the operator in the left-hand side originally acts on column vectors of numbers while the operator in the right-hand side acts on functions in the variable $x$. The above equality means that their extended action on the sequence $\mathcal P$ of monic eigenpolynomials coincide. 

\smallskip
Notice that \eqref{eq:impeq} can be also proven directly by calculating its left-hand and right-hand sides explicitly. Namely, both  give the following expression 
$$
(\la_n-\la_{n+1})P_{n+1}(x)+\sum_{j=0}^d b_j(n)(\la_n-\la_{n-j})P_{n-j}(x)
$$
for the $n$-th coordinate in the image. 

\medskip
 Formula~\eqref{eq:impeq} implies that for $j=1,2,\dots$, 
 \begin{equation}\label{eq:almost}
 \ad^j_x L(\mathcal P(x)) = (-1)^j\ad^j_{\La}\mathcal L(\mathcal P(x)).
 \end{equation} 
 
Now the identity  \eqref{eq:j4} together with \eqref{eq:almost} imply 
 
 $$\ad^4_{\La}\mathcal L \equiv 0$$
 since the left-hand side of \eqref{eq:almost} for $j=4$ is vanishing identically and $\ad^4_{\La}\mathcal L$  is a difference operator whose kernel contains number sequences $\mathcal P(x)$ for every fixed real $x$.  
 \smallskip
 (Formula \eqref{eq:almost} can be obtained by explicit calculations similar to that following \eqref{eq:impeq}, but we omit them here).
 \qed
 
 \begin{remark}
{\rm  Equation \eqref{ad} will be referred to as the \emph {$\ad$-condition} and will be our main technical  tool. Appendix I below contains a purely computational proof of Lemma~\ref{lm2.1}  which we decided to include  after  a number of questions raised by an anonymous referee.}
 \end{remark} 

 \medskip\noindent
{\em Notation.} Below we will impose the restriction that $\La$ is a monic difference operator of order 3, i.e.,
\begin{equation}\label{Lambda}
(\La)_n = T + \sum\limits_{i=0}^{2}b_i(n)T^{-i},
\end{equation}
 comp. Conjecture~\ref{conj:order}.  In what follows we will use the following notation for the coefficients of the polynomials $a_3(x), a_2(x)$ and $a_1(x)$ respectively:
 $$
 \begin{cases}
 a_3(x)=a_{33}x^3+a_{32}x^2+a_{31}x+a_{30};\\
 a_2(x)=a_{22}x^2+a_{21}x+a_{20};\\
  a_1(x)=a_{11}x+a_{10}.\\
 \end{cases}
 $$

\subsection{ The ad-condition}

Notice that if $L$ is a differential operator of order $3$ then the sequence $\mathcal L=\{\la_n\}_{n=0}^\infty$ is the sequence of values of a polynomial  of degree  at most $3$ in the variable $n$ at consecutive non-negative integer points. Indeed,  the entries of $\la_n$ come from the differentiation of  the term $x^n$ in $P_n(x)$. More exactly,
$$\la_n=(n)_3\cdot a_{33}+(n)_2\cdot a_{22}+n\cdot a_{11}.$$

\smallskip

\ble{lemma:ad3:poly}

The operator $\ad^3_\La \mathcal L$ is of the form:
\beq \label{ad3-1}
\ad^3_\La \mathcal L =\al\La^3+  \beta \La^2 + \gamma \La + \delta I = 6a_3(\La),
\eeq
where the coefficients $\al, \be, \ga, \de \in \Cset$ are independent of $n$.
\ele

\proof

From the condition
$\ad^4_{\La} \mathcal L\equiv 0$ it follows that $\ad^3_{\La} \mathcal L$ is a function of the difference operator $\La$ which we denote by $q(\La)$. Moreover in our situation $q(\La)$ is a polynomial of  degree  at most 3 in $\La$ since it contains $T$ in  at most the third degree.   Also notice that
\[
\ad^3_xL =(-1)^36a_3(x).
\]

Having in mind that
\[
\ad^3_xL( \mathcal P(x)) = (-1)^3 \ad_{\La}^3\mathcal L (\mathcal P(x)),
\]
we obtain
\[
(-1)^36a_3(\La)(\mathcal P(x)) = (-1)^3 \ad_{\La}^3\mathcal L (\mathcal P(x)),
\]
i.e., $q(t)  =6 a_3(t).$
\qed

 \medskip
 \begin{remark} {\rm Formula~\ref{ad3-1} says that the leading polynomial coefficient $a_3(x)$ coincides with $\frac{1}{6}(\al x^3+\be x^2 +\ga x+\de)$, i.e., in the above notation $a_{33}=\al/6,\;a_{32}=\be/6, a_{31}=\ga/6,\; a_{30}=\de/6$ which we will  use below.}

 \end{remark}
 
  \begin{remark} {\rm Observe that if two finite-dimensional linear operators commute and one of them has a simple spectrum then the other is a (polynomial) function of the first one. This fact  is unknown and probably false for commuting infinite-dimensional operators. However the latter Lemma provides such a statement in our particular case.
  
  }
  \end{remark} 

 \section{Case when $\mathcal L $ is given by a linear polynomial}
As we mentioned above, the sequence $\mathcal L=\{\la_n\}_{n=0}^\infty$ is given by  a polynomial in $n$ of degree at most $3$. Here we will discuss  the simplest case when $\la_n$ is  given by a linear polynomial. In terms of the differential operator $L$ given by \eqref{eq:L},  the linearity of $\la_n$ means that  $\deg a_3(x) <3$ (which is equivalent to $a_{33}=\al/6=0$),  $\deg a_2(x)<2$ (which is equivalent to $a_{22}=0$) and $\deg a_1(x)=1$ (which is equivalent to $a_{11}\neq 0$). Additionally, observe that in this case we can without loss of generality assume that $\la_n \equiv n$.  Indeed, if  $\la_n = \mu n + \nu,$ then we can consider $L - \nu$ instead of  $L$. Finally, we can divide both sides by $\mu$, which results in  $\la_n \equiv n$.

\medskip
Assuming that $\la_n \equiv n$, we will separately study the following  3 (sub)cases:

(i) $a_3(x)$ is a quadratic polynomial, i.e., $a_{33}=0; a_{32}=\be/6\neq 0$;

(ii) $a_3(x)$ is a linear polynomial, i.e., $a_{33}=a_{32}=0$, $a_{31}=\ga/6\neq 0$;

 (iii) $a_3(x)$ is a non-vanishing constant, i.e., $a_{33}=a_{32}=a_{31}=\ga/6=0$, $a_{30}=\de/6 \ne 0$.

\subsection{Case $\deg a_3(x)=2 \Leftrightarrow \be \neq 0$}\label{sub:3.1}

Since $\la_n  \equiv n$, we get
\[
L = (a_{32} x^2 + a_{31} x +a_{30})\partial_x^3 +(a_{21}x +a_{20})\partial^2 + (x +a_{10})\partial_x,
\]
where $a_{32}\neq 0$.
Rescaling $L$ we can achieve $a_{32} =1$.   Shifting $x$ we can additionally assume that $a_{10}=0$. (Both changes  only result in somewhat shorter formulas, but otherwise are not essential.) Thus without loss of generality, we can assume that  the operator $L$ is of the form
\beq  \label{L}
L = ( x^2 + a_{31} x +a_{30})\partial_x^3   +(a_{21}x +a_{20})\partial^2 + x \partial_x.
\eeq

To prove the required result, we do the following. Expanding the polynomial $P_n(x)$  as
\[
P_n(x)  = x^n + p_1(n)x^{n-1} + \ldots + p_k(n)x^{n-k} + \ldots,
\]
we are going to compute the coefficients $p_k(n)$ for the first few values of $k$. Then  we  compute the coefficients $b_j(n)$ of the recurrence relation $\La$. The condition that the polynomials $P_n(x)$ satisfy a 4-term relation implies that $b_3(n) \equiv 0$. We will use the following identities:
\beq \label{a-j}\begin{cases}
\begin{split}
	p_1(n) &= p_1(n+1) + b_0(n);\\
	p_2(n) &= p_2(n+1) + b_0(n)p_1(n) +b_1(n);\\
	p_3(n) &= p_3(n+1) + b_0(n)p_2(n) +b_1(n)p_1(n-1)   +b_2(n);    \\
	p_4(n) &= p_4(n+1) + b_0(n)p_3(n) +b_1(n)p_2(n-1)   +b_2(n) p_1(n-2) + b_3(n).
\end{split}
\end{cases}
\eeq

These relations imply the following expressions for $b_j(n)$:
\beq  \label{b-j}\begin{cases}
\begin{split}
	b_0(n) &= p_1(n) - p_1(n+1); \\
	b_1(n)& =   p_2(n)-  p_2(n+1) - b_0(n)p_1(n);\\
	b_2(n)  & = p_3(n)-   p_3(n+1) -b_0(n)p_2(n) - b_1(n)p_1(n-1);   \\
	b_3(n)  &= p_4(n)  -    p_4(n+1) - b_0(n)p_3(n) -b_1(n)p_2(n-1)  -b_2(n) p_1(n-2).
\end{split}\end{cases}
\eeq

We want to conclude that under our assumptions, we get $a_{31} = a_{30} =0$.
The following statement being the major result of this section claims exactly this.

\bpr{lin}  If a differential operator $L$ of the form
\beq  \label{LL}
L = ( x^2 + a_{31} x +a_{30})\partial_x^3   +(a_{21}x +a_{20})\partial^2 + x\partial_x
\eeq
has a sequence of polynomial eigenfunctions  $\{P_n(x)\},  \; n=0, 1,\ldots$ which  satisfies a 4-term recurrence relation, then $a_{31}= a_{30} = a_{20} =0$.  Thus the operator $L$ will have the form $\rm 1)$ from Theorem \ref{th:main}, i.e., 
\[
L = \sum_{j=1}^{3}a_j x^{j-1}\partial^j +x\partial.
\]
\epr

\proof

We will actually show that if $b_3(n) \equiv 0$ (which is a necessary condition for the sequence $\{P_n(x)\}$ to satisfy a 4-term recurrence), then  all three parameters $a_{31}, a_{30}, a_{20}$ vanish. To do this, we will not need the  explicit form of $b_3(n)$ which is very cumbersome, but it would be enough to obtain  polynomial coefficients at $a_{31}$,   $a_{30}$   and $a_{20}$ in the presentation of $b_3(n)$.

The reason for the latter sufficiency  is as follows. It has been shown in  \cite{Ho1}  that if $a_{31} = a_{30} = a_{20} =0$, then $b_3(n) \equiv 0$. Below we will refer to $a_{31},  a_{30},   a_{20}$   as \textit{parameters}.

We will calculate several terms in the expansion of $b_3(n)$ (as well as other $b_i(n)$'s) in powers of parameters with coefficients being powers of $n$. Namely, we denote   by $``\sim "$ the following $2$-step truncation of the expansion of $b_i(n)$ in the latter powers. Firstly, we discard all terms containing powers of parameters exceeding $1$. Secondly, for each parameter, we only keep the highest power of $n$ in its coefficient. For example, 

\[
b_3(n)  \sim   n^{m_1} \cdot a_{31} + n^{m_2}\cdot a_{30}   + n^{m_3} \cdot a_{20},
	\]
	where the numbers $m_j$ are some nonnegative integers   and the relation itself means that $b_3(n) $ is  equal to the r.h.s. modulo some terms that contain    higher powers of  parameters.   Here $n^{m_j}$ is the leading power of $n$ in the coefficient of the corresponding parameter. 

\medskip
  To compute these leading terms,  we first calculate the corresponding terms  in $4p_4(n)$.  Namely, set
$$4p_4(n)   = c n^{q_0} +  a_{31} n^{q_1}  + a_{30}n^{q_2}   + a_{20} n^{q_3},$$
    where $c \in \Cset$  does not depend on the parameters.
 (Later we are going to use a similar expansion for  $b_3(n)$).

\medskip
Using the formulas \eqref{b-j} and 
applying the operator \eqref{L}  to $P_n(x)$,  we get
\begin{eqnarray*}
 &L(P_n) &=( x^2 + a_{31} x +a_{30})  ((n)_3\cdot x^{n-3}    +   p_1(n)(n-1)_3\cdot x^{n-4} +\ldots +p_k(n)(n-k)_3\cdot x^{n-k-3}+\ldots)\\
&+&(a_{21}x +a_{20})((n)_2\cdot x^{n-2}    +   p_1(n)(n-1)_2\cdot x^{n-3}+ \ldots +p_k(n)(n-k)_2\cdot x^{n-k- 2}+\ldots)\\
&+& x(n\cdot x^{n-1}    +   p_1(n)(n-1)\cdot x^{n-2}+ \ldots +p_k(n)(n-k)\cdot x^{n-k-1}+\ldots)\\
&=& n(x^n +p_1(n)x^{n-1}  +\ldots + p_k(n)x^{n-k} +\ldots).
\end{eqnarray*}

Comparing the coefficients of  $x^{n-1}$ in  both sides of the latter equality, we obtain

\[
p_1(n)  =  (n)_3  +(n)_2\cdot a_{21}.
\]

From \eqref{b-j} 
we obtain
\[
b_0(n)  = - \Delta (p_1(n))  =-3(n)_2 - 2n\cdot a_{21}
\]
which means that both $p_1(n)$  and $b_0(n)$ do not contain our parameters.  We observe that in order to calculate  $b_3(n) $,   we  need  first to compute  $b_1(n)$  and   $b_2(n)$.

Starting with $p_2(n)$ we obtain
\[
2p_2(n)  = (n)_3 \cdot a_{31}+  (n)_2\cdot a_{20} + p_1(n-1)p_1 (n).
\]
For later use,  notice that
\[
p_2(n)  \sim   \frac{1}{2} \left[ n^3 \cdot a_{31}+  n^2\cdot a_{20} + n^6\right].
\]

Using \eqref{b-j} again, we get
\[
\begin{split}
b_1(n) &=      -\Delta (p_2(n))     - b_0(n)p_1(n) \\
&  =  -\frac{3(n)_2}{2}\cdot a_{31}  - n \cdot a_{20}  +    \frac{1}{2}p_1(n)\Delta^2 (p_1(n-1))
\end{split}
\]
which gives
\[
b_1(n)  \sim  -\frac{3n^2}{2} \cdot  a_{31}   -   n\cdot a_{20}    +  3 n^4.
\]

Let us  now apply a similar procedure to find $p_3(n)$.  From the above expression for $L(P_n(x))$  we get
\[
\begin{split}
3p_3 (n)&=  (n)_3\cdot a_{30} +   p_1(n) (n-1)_3\cdot a_{31} +   p_1(n)(n-1)_2\cdot a_{20}+  p_1(n-2)  p_2(n)\\
&=  \frac{1}{2}p_1(n-2)\left((n)_3\cdot a_{31} +  (n)_2 \cdot a_{20}+ p_1(n-1)p_1 (n)\right)\\
&  +(n)_3 \cdot  a_{30} +  p_1(n) (n-1)_3 \cdot a_{31}  +   p_1(n)(n-1)_2\cdot a_{20}.
\end{split}
\]

The latter relation implies that
\[
p_3 (n)    \sim   \frac{n^3}{3}\cdot a_{30} +   \frac{n^6}{2}\cdot a_{31}   + \frac{n^5}{2} \cdot  a_{20}
  + \frac{1}{3!} n^9.
\]

\medskip
For $b_2(n)$,  we get
\[
b_2(n) = -  \Delta (p_3(n))     - b_0(n) p_2(n)  - b_1(n)p_1(n)
\]
which implies that
\[
b_2(n)   \sim  n^4\cdot  a_{31} -n^2 \cdot a_{30}+\frac{2n^3}{3}\cdot a_{20}.
\]

Further, computing $p_4(n)$, we get
\[
4p_4(n)    =    p_2(n) (n-2)_3\cdot a_{31} +    p_1(n)(n-1)_3 \cdot a_{30}
  + p_2(n)(n-2)_2\cdot a_{20} +           p_1(n-3)\cdot p_3(n)
\]
which implies that
\[
p_4(n) \sim   \frac{n^9}{4}\cdot a_{31} +    \frac{n^6}{3}\cdot a_{30} + \frac{n^8}{4}\cdot a_{20} + \frac{1}{4!}n^{12}.
\]

Finally, we obtain
\[
b_3(n)   =  -\Delta (p_4(n)) - b_0(n)p_3(n) - b_1(n)p_2(n) - b_2(n) p_1(n)
\]
by plugging the expressions  for $p_j(n), b_j(n)$    from    the  above formulas. This gives
\[
b_3(n)  \sim  \frac{9n^7}{2}\cdot a_{31}+\frac{9n^4}{2}\cdot a_{30}+{3n^6}\cdot a_{20}. 
\]
As the consequence of the latter expansion, we see that if we assume that $b_3(n) \equiv 0$, then, in particular,   we get that $a_{31} = a_{30} = a_{20} =0$. \qed


\subsection{Case $\deg a_3(x) \leq1 \Leftrightarrow a_{32} = 0$}
\label{sec:beta_is_0}

In this subsection our goal  is to prove the following statement.

\bpr{le:ad3beta_is_0}\label{th:imp}
Solution of Problem 2 for $\deg a_3(x) \leq1$ is possible, if and only if  $a_{31} = a_{21} = 0$  and  $a_{30}\neq 0$.
In this case we obtain the Appell polynomials as eigenpolynomials of $L$.
\epr

\proof

We use the same approach as in Subsection~\ref{sub:3.1}.  The operator $L$ is of the form
\[
L = ( a_{31} x +a_{30})\partial_x^3 +(a_{21}x +a_{20})\partial^2_x + (x +a_{10})\partial_x.
\]
By  translation of $x$ we can make $a_{10}= 0$ and work with
\[
L = ( a_{31} x +a_{30})\partial_x^3 +(a_{21}x +a_{20})\partial^2_x + x\partial_x.
\]
(Again such change of variables  only results in somewhat shorter formulas, but otherwise is not essential.)

Expanding 
\[
P_n(x)  = x^n + p_1(n)x^{n-1} + \ldots + p_k(n)x^{n-k} + \ldots,
\]
let us apply $L$ to the polynomial  $P_n(x)$. Straightforward calculation gives for $L(P_n)$ the expression
\begin{eqnarray*}
	&L(P_n) &=( a_{31} x  +a_{30})  ((n)_3\cdot x^{n-3}    +   p_1(n)(n-1)_3\cdot x^{n-4}+ \ldots +p_k(n)(n-k)_3\cdot x^{n-k-3}+\ldots)\\
	&+&(a_{21}x +a_{20})((n)_2\cdot x^{n-2}    +   p_1(n)(n-1)_2\cdot x^{n-3} +\ldots +p_k(n)(n-k)_2\cdot x^{n-k- 2}+\ldots)\\
	&+& x(nx^{n-1}    +   p_1(n)(n-1)x^{n-2}+ \ldots +p_k(n)(n-k)x^{n-k-1}+\ldots)\\
	&=& n(x^n +p_1(n)x^{n-1}  +\ldots + p_k(n)x^{n-k}+ \ldots).
\end{eqnarray*}

Comparing the coefficients of   $x^{n-1}$ in both sides,  we find
\[
p_1(n)  =  (n)_2\cdot a_{21}.
\]
From the expansion of $\La$ we conclude that
\beq \label{b0}
b_0(n)  = -2n\cdot a_{21}.
\eeq

From the coefficients of $x^{n-2}$  we get
\[
p_2(n)  = \frac{1}{2}  \left( (n)_3 \cdot a_{31} +     (n)_2\cdot a_{20}  +p_1(n)p_1(n-1)\right).
\]

Evaluating $b_1(n),b_2(n),b_3(n)$ from \eqref{b-j} we arrive at
\begin{equation}\label{b1}
b_1(n)  = \frac{1}{2} n \left((n-1) \left(2 a_{21}^2-3 a_{31}\right)-2 a_{20}\right), 
\end{equation}

\begin{equation}\label{b2}
b_2(n)  = \frac{1}{3} n (n-1) \left(2 (n-2) a_{21} a_{31}-3 a_{30}\right), 
\end{equation}

and 

\begin{equation}\label{b3}
b_3(n)  = \frac{1}{8} n \left(n^2-3 n+2\right) \left(2 (n-3) a_{31} a_{21}^2+a_{31}
   \left(3 (n-3) a_{31}+4 a_{20}\right)-8 a_{30} a_{21}\right).
\end{equation}

The second factor in the right-hand side of \eqref{b3} equals $a_{31}(n-3)(2 a_{21}^2+3a_{31})+(4a_{31}a_{20}-8 a_{30} a_{21})$.
For $b_3(n)\equiv 0$  as a polynomial in $n$ one has the following two options  
$$\begin{cases} a_{31}=0\\4a_{31}a_{20}-8 a_{30} a_{21}=0\end{cases} \quad \text{or} \quad \begin{cases} 2 a_{21}^2+3a_{31}=0\\4a_{31}a_{20}-8 a_{30} a_{21}=0\end{cases}.$$

To refute the second option we have to involve the expressions for $b_4(n)$ which is as follows:

\begin{equation}\label{b4}
b_4(n)=\frac{1}{15} a_{21} (n-3) (n-2) (n-1) n \left(2 a_{31} a_{21}^2 (n-4)+2 a_{31}
   \left(3 a_{31} (n-4)+5 a_{20}\right)-15 a_{30} a_{21}\right).
   \end{equation}
   
   One can check that the equations of the second option lead to $b_4(n)\not \equiv 0$ which is a contradiction. Finally, the first option splits into two possibilities: $a_{31}=a_{21}=0$ and $a_{31}=a_{30}=0$.


\medskip
The case $a_{31} =a_{21}=0$ leads to  $b_3(n)\equiv 0$ as claimed above. (The case $a_{31}=a_{30}=0$ is forbidden since then $L$ becomes a second order differential operator.) 
\qed

 \section{Case when $\mathcal L$ is given by a quadratic polynomial}\label{sec:quadr}

 This case corresponds to $\deg a_3(x)<3$, $\deg a_2(x)=2$, and $\deg a_1(x)\le 1$. We want to show that in this situation there are no differential operators whose eigenpolynomials satisfy a finite recurrence relation.

\bpr{quad}
 The case when $\la_n$ is quadratic is impossible, i.e., no linear differential operator $L$ of order $3$ satisfying this condition can solve Problem~\ref{prob:genBK}.

 \epr
 \proof

 Our arguments are partially computer-aided due to the complexity of calculations. (The corresponding Mathematica code and its results can be requested from the third author). 
 The scheme of what we are doing is presented below.

 We are going to compute the coefficients of the 4-term recurrence in terms of the coefficients of the operator $L$.
 In the case under consideration,  we can write the  operator $L$ in the form

 \[
 L = (a_{32} x^2 + a_{31} x +a_{30})\partial^3  +(x^2 +a_{21}x +a_{20})\partial^2 + (a_{11} x+a_{10})\partial.
 \]

We can also assume that $a_{21}  =0$, which can be achieved by a translation of $x$. (By a slight abuse of notation we use the same letters for the coefficients of $L$.) From the above form of $L$ we get that $$\la_n=n(n-1)+\nu n.$$

 As above, introduce 
 \[
 P_n(x)  = x^n +p_1(n)x^{n-1}+ \ldots + p_k(n)x^{n-k}+ \ldots .
 \]

We are going to compute the coefficients $p_k(n)$ of  $P_n(x)$  by  considering  $L(P_n(x))$   and  using the formulas \eqref{a-j}   and \eqref{b-j}.   $L(P_n(x))$ satisfies the relation
\begin{eqnarray*}
	&L(P_n)= &(a_{32}  x^2 + a_{31} x +a_{30})  ((n)_3x^{n-3}    +   p_1(n)(n-1)_3x^{n-4} \ldots +p_k(n)(n-k)_3x^{n-k-3}+\ldots)\\
	&+&(x^2 +a_{20})((n)_2x^{n-2}    +   p_1(n)(n-1)_2x^{n-3} \ldots +p_k(n)(n-k)_2x^{n-k- 2}+\ldots)\\
	&+& (a_{11} x+a_{10})(nx^{n-1}    +   p_1(n)(n-1)x^{n-2} \ldots +p_k(n)(n-k)x^{n-k-1}+\ldots)\\
	&=&  (n(n-1) +\nu n)(x^n +p_1(n)x^{n-1}  +\ldots + p_k(n)x^{n-k} +\ldots).
\end{eqnarray*}

We now equalize the coefficients at the same powers of $x$ in the right-hand  and the left-hand sides of the latter equation. For our purposes it would be enough to find expressions for $p_1(n), p_2(n), p_3(n)$, and $p_4(n)$ only. Knowing $p_1(n), p_2(n), p_3(n), p_4(n)$ and using the above formulas, we can express $b_0(n), b_1(n), b_2(n)$ and $b_3(n)$.  As before, to obtain a recurrence relation of order at most $4$ for the eigenpolynomials, we need to ensure that $b_3(n)\equiv 0$.

We will normalize the above expression for $L$ using the action of the affine group on $x$. First, let us consider the case when  $a_{32}\neq 0$.  By rescaling we can obtain $a_{32}  = 1$.   Then the leading coefficient of the polynomial of degree $9$ in the numerator of  $b_3(n)$ equals $8$ which implies that $b_3(n)$ can not vanish identically, see \eqref{b3}. 



If $a_{32} =0$, but  $a_{31} \neq 0$  by rescaling we can assume that $a_{31} = 1$. In this case the leading coefficient of the polynomial of degree $7$ in the numerator of $b_3(n)$ equals $24$ which again implies that $b_3(n)$ can not vanish identically. 

Next assume that $a_{32} = a_{31} = 0 $.   Let us rescale $x$  to obtain   $a_{30} =1$.  

 We still need to find if there exist  values of $  a_{20},$ $ a_{11}, a_{10}$ for which $m_9=m_8=m_7=m_6=0$, where $m_i$ is the coefficient of $n^i$ in the expression     \eqref{a20} for $b_3(n)$.     
 
 First, we get that $m_9=8(96a_{20}^2+288a_{10})$. 
 Assuming that $m_9=0$ we obtain $a_{10}=-\frac{a_{20}^2}{3}$. Inserting the latter expression for $a_{10}$ in $b_3(n)$, and using $m_8=0$, we get  that either $a_{11}=2$ or $a_{20}=0$. In the former case, setting $a_{10}=-\frac{a_{20}^2}{3}$ and $a_{11}=2$ in the expressions for $m_7$ and $m_6$, we get
  $$\begin{cases}
 m_6=-\frac{64}{3}a_{20}^2(27+4a_{20}^3);\\ 
 m_7=\frac{64}{3}a_{20}^2(9+2a_{20}^3).
 \end{cases}
 $$
  Thus, if $a_{20}\neq 0$, then the equations for 
 $m_6=0$ and $m_7=0$ are obviously
 incompatible, i.e., have no common solutions.

Now let us consider the case when $ a_{32} = a_{31} =a_{21} = a_{20}=0$.  We will show that $b_3(n) \equiv 0$, which however is not enough to claim that all $b_j (n) \equiv0$ for $j > 2$. For this reason we are going to compute them up to $j=5$.     We know that in this case also $a_{10} =0$. Again applying the operator $L$ to the polynomials $P_n(x)$,  we obtain

\begin{eqnarray*}
	&L(P_n)= &  (n)_3\cdot x^{n-3}    +   p_1(n)(n-1)_3\cdot x^{n-4} \ldots +p_k(n)(n-k)_3\cdot x^{n-k-3}+\ldots\\
	&+&x^2((n)_2\cdot x^{n-2}    +   p_1(n)(n-1)_2\cdot x^{n-3} \ldots +p_k(n)(n-k)_2\cdot x^{n-k- 2}+\ldots)\\
	&+& a_{11} x(nx^{n-1}    +   p_1(n)(n-1)x^{n-2} \ldots +p_k(n)(n-k)x^{n-k-1}+\ldots)\\
	&=&  (n(n-1) +  a_{11} n)(x^n +p_1(n)x^{n-1}  +\ldots + p_k(n)x^{n-k} +\ldots).
\end{eqnarray*}

Then, for $k>3$,  we find the following formulas for $p_j(n)$:
\[
\begin{cases}
p_1(n)  (2 -2n -a_{11})& = 0;\\
p_2(n)(6-4n-2a_{11})   &=   0; \\
p_3(n)(12-6n-3a_{11})   &=  - (n)_3; \\
 \vdots\quad \quad \quad\vdots \quad \quad \quad\vdots \quad\quad \quad\vdots  &  \quad\quad \vdots  \\ 
p_k(n)(k(k+1)-2kn-ka_{11})   &=     - (n-k+3)_3\cdot p_{k-3}(n). \\
\end{cases}
\]
By induction, we see that $p_{3j}(n) \not{\!\!\equiv} 0$\; while  $p_i (n)\equiv 0$ for $i\neq 3j$.   In particular, 
\[
p_6(n)(42 -12n - 6a_{11})   = - (n-3)_3 \cdot p_{3}(n) 
\]
i.e.,
$$
p_6(n) = \frac{(n)_5 }{18(2n+a_{11}  -4)   (2n+a_{11} - 7) }.
$$

 From this we easily compute  several first $b_j(n)$'s getting
\[
\begin{cases}
b_0(n)& = 0;  \\
b_1(n) & = 0;  \\
b_2(n)  &= - \Delta p_3(n) \neq 0; \\
b_3(n)  &= - \Delta p_4(n) -b_0(n)  p_3(n)  - b_1(n) p_2(n)  -b_2(n) p_1 (n) =0; \\
b_4(n)  &= - \Delta p_5(n) -b_0 (n) p_4(n)  - b_1(n) p_3(n)  -b_2(n) p_2(n)   -b_3(n) p_1(n)   =0; \\
b_5(n)  &= - \Delta p_6(n) -b_0(n)  p_5(n)  - b_1(n) p_4(n)  -b_2(n) p_3(n)   -b_3(n) p_2(n)  -b_4(n) p_1(n). \\
\end{cases}
\]
In the expression for $b_5(n) $ there are two non-vanishing   terms:  $\Delta p_6(n)$  and $b_2(n) p_3(n) $.      Both terms are of order $n^3$. We only need to show that  their sum is not identically zero. Straight-forward calculations give 
\[
b_2(n) p_3  (n) =  \frac{1}{9} \frac{(n)_3 \cdot (n)_2\cdot  (4n - a_{11})}{(2n +a_{11}- 2)(2n +a_{11}-4)^2}
\]
and
\[
\Delta p_6(n) = \frac{(n)_4\cdot R(n)}{(2n +a_{11}- 2)(2n +a_{11}-4)(2n +a_{11}- 5)(2n +a_{11}- 7)       }, 
\]
 where $R(n)$ is a polynomial of degree at most $2$, whose explicit expression is irrelevant for our purposes.
We see that the sum of $b_2(n) p_3(n) $  and $\Delta p_6(n)$ cannot vanish since they contain different factors.

In fact, the explicit expression for $b_5(n) $ can be obtained even without  computer. Our symbolic computations give

\begin{equation}\label{b5}
b_5(n)=\frac{n(n-1)(n-2)(n-3)(n-4) \left(3 a _{11}+5 n-16\right)  }{9
	\left(a _{11}+2 n- 2\right) \left(a _{11}+2 n-7\right) \left(a _{11}+2
	n-5\right) \left(a _{11}+2 n-4\right)}.
\end{equation}
\qed


 \section{Case when $\mathcal L$ is given by a cubic polynomial}
 When $\mathcal L=\{\la_n\}$ is given by a cubic polynomial in $n$, by using translation and scaling, we can without loss of generality assume that
 $$\la_n=n(n-1)(n-2)+\nu n(n-1)+\mu n.$$
 As before, one can easily observe that $\nu=a_{22}$ and $\mu=a_{11}$.

 \medskip
 Using this ansatz, we first formulate some important preliminary statements about the coefficients of $\La$ and $L$.

 \ble{T3}
 The coefficient $\al$ in \eqref{ad3-1} equals 6.
 \ele
 \proof

 To prove this, we compute the term of the highest degree in $T$ in the two expressions for $\ad^3_{\La}\mathcal L$ which we presented earlier. On one hand,  from \eqref{ad3-1} we obtain that $(\ad^3_{\La}\mathcal L)_n = \al T^3+ \dots$.
 Let us compute the term containing    $T^3$ in $(\ad^3_{\La}\mathcal L)_n$.      Simple computation shows that

\[\begin{cases}
\begin{split}
(\ad^1_{\La} \mathcal L)_n&=  (\Delta (\mathcal L))_n \cdot T +\ldots ; \\
 (\ad^2_{\La} \mathcal L)_n&=  (\Delta^2(\mathcal L))_n\cdot T^2 +\ldots ;\\
 (\ad^3_{\La} \mathcal L)_n&=  (\Delta^3(\mathcal L))_n\cdot T^3 +\ldots .
\end{split}
\end{cases}
 \]
 Since $\Delta^3 (\mathcal L) \equiv 6,$ the result follows.
 \qed

 \ble{T-6}\label{lm:T-6}
 The coefficient   $b_2(n)$ in the expansion \eqref{Lambda} of $\La$  vanishes identically.

 \ele

 \proof

 Suppose that $b_2(n)\not \equiv 0$. Then again we can use the above approach and   compute  the term in  $(\ad^3_{\La}\mathcal L)_n$ containing $T^{-6}$ in two different ways. On one hand, if taken from     $6\La^3$,   this term  is given by
 \[
 \al\cdot b_2(n)b_2(n-2)b_2(n-4)T^{-6} = 6 b_2(n)b_2(n-2)b_2(n-4)T^{-6}.
 \]

 On the other hand,  the computation of  the same term   from      $(\ad^3_{\La}\mathcal L)_n$ shows that it coincides with
 \[
 b_2(n)b_2(n-2)b_2(n-4)(- \la_{n-6}  +3\la_{n-4}-3\la_{n-2} +\la_n )T^{-6}.
 \]
 However,  $- \la_{n-6}  +3\la_{n-4}  -3\la_{n-2} +\la_n = 48$ which is only possible  if $b_2(n) \equiv 0$. Indeed, since $b_2(n)$ is a polynomial in the variable $n$,  the equation $$ b_2(n)b_2(n-2)b_2(n-4)=0$$ can not have infinitely many solutions unless $b_2(n) \equiv 0$. 
  \qed

 \ble{b1}\label{lm:T-1}
 Under our current assumptions, the coefficient   $b_1(n)$ in the expansion \eqref{Lambda} of $\La$ vanishes identically.
 \ele

 \proof
 The argument is similar to the above computation  of $b_2(n)$. Assume that $b_1(n)\not \equiv 0$. Notice that by Lemma~\ref{lm:T-6},
  \[
 (\La)_n = T + b_0(n)T^0+ b_1(n)T^{-1}.
 \]
 The  coefficient of $T^{-3}$ in the expression for $6\La^3$ is given by
$ 6 b_1(n)b_1(n-1)b_1(n-2). $

From the expression for     $(\ad^3_{\La}\mathcal L)_n$,   we find that the same term coincides with
  \[
 b_1(n)b_1(n-1)b_1(n-2)(\la_{n-3} -3\la_{n-2} +3\la_{n-1} -\la_n) =     -6b_1(n)b_1(n-1)b_1(n-2).
 \]
Again we get a contradiction.
 \qed

 \medskip

 \ble{new-ad}\label{lm:T-0}
 The coefficient  $b_0(n)$ in the expansion \eqref{Lambda} of $\La$ is an arbitrary constant.
  \ele

 \proof By Lemmas~\ref{lm:T-6} and~\ref{lm:T-1},   $\La = T +b_0(n)$ which makes the computation of both $\La^3$ and $\ad_{\La}^3 \mathcal L$ very easy. Moreover, we only need  the coefficient at $T^0=Id$.    We have
 \[
 (\La^j)_n = \ldots +b_0^j(n)T^0,  \; j=0,1,\dots.
 \]
For $(\ad_{\La}^j\mathcal L)_n,\; j=0,1,\dots,$ it is obvious that it does not contain $T^0$.
 Comparing the coefficients at $T^0$, we obtain
 \[
 \al b_0^3(n)  + \be b_0^2(n)  + \ga b_0(n)  + \de = 0.
 \]
 However this is only possible if $b_0(n) $ is a constant.
 
\qed

Denoting this above constant by $p$ we get the following claim. 
\medskip
\bpr{cubic}
The only linear differential operators of order $3$ solving Problem~\ref{prob:genBK} which have polynomial eigenfunctions with $\la_n = n(n-1)(n-2) +\nu n(n-1) + \mu n $ are of the form
\[
L = 6(x-p)^3\partial_x^3 + \nu (x-p)^2\partial_x^2 + \mu (x-p) \partial.
\]
Hence any such operator $L$ is reducible.
\epr

\proof  By the last lemma we have

\[
xP_n(x)  = P_{n+1}(x) + pP_n(x).
\]
Hence $P_{n+1}    =  (x-p)^n$. The corresponding differential operator of order 3 is as claimed.

\qed

\section{Final Remarks}

\noindent

\noindent
{\bf I.}  To the best
of our knowledge, Conjecture~\ref{conj:genBK}
 containing a complete conjectural description of the set of all solutions to the algebraic version of the classical Bochner-Krall Problem~\ref{CBK-problem} and its generalisation Problem~\ref{prob:genBK} has never previously appeared in the literature.  
  However the way our conjectures are formulated, it is difficult to verify them for operators of somewhat high order unless one finds an alternative description. On the other hand, the  special case of operators of order $4$ seems to be doable by using the methods of the present paper. 

\medskip
\noindent
{\bf II.}   One additional restriction of the validity of the results obtained in the present paper comes from the fact that we are using the assumptions  of  Conjecture~\ref{conj:order}. To actually claim that we have solved Problem~\ref{prob:genBK}  for all operators of order $3$, we need to settle this conjecture at least in this case. We hope to return to this project in a future publication. 

\medskip
\noindent
{\bf III.} Our proof of the main Theorem~\ref{th:main} is based on consideration of a large number of special subcases and is partially computer-aided. Such method is not very effective in order to be able to approach our general conjectures  and a more conceptual understanding of our proof is needed. 

\medskip
\noindent
{\bf IV.} In connection with Proposition~\ref{prop:translate} and Wendroff's theorem we want to ask under which conditions on the sequences $\{u_n\}$ and $\{v_n\}$ the polynomial sequence $\{P_n(x)\}$ consists of real-rooted polynomials.

\medskip
\section {Appendix I. Proving Lemma~\ref{lm2.1} by explicit computations}
Here us present an alternative approach to evaluation of \eqref{eq:j4}. To this end, we
find explicit form of $\left(\Lambda\right)_n$, i.e. we express $b_j(n)$, $j=0,1,2$
in terms of the coefficients $a_{ij}$ of the differential operator $L(x,\partial)$. The operator $\left(\Lambda\right)_n$ is
represented as a four-diagonal infinite matrix and $\mathcal{L}$ as an infinite matrix
with nonzero main diagonal $\{\mu_n\}$.

\smallskip
More explicitly, one has
$$
\mu_n=n a_{11}+(n-1) n a_{22}+(n-2) (n-1) n a_{33}\quad 
\text{and} \quad 
b_0(n)=\rho/\sigma,
$$
where
\begin{equation*}
\begin{aligned}
\rho=&-3 (n-2) n (n-1)^2 a_{32} a_{33}-2 n (n-1) a_{21} a_{22}-3 n (n-1)
   a_{11} a_{32}- \\
   &2 n (2 n-1) (n-1) a_{22} a_{32}+3 (n+2) (n-1) a_{10}
   a_{33}+6 n (n-1) a_{21} a_{33}- \\
   &2 n a_{11} a_{21}-a_{10} a_{11}+2 a_{10} a_{22};
\end{aligned}
\end{equation*}
\begin{equation*}
\sigma=\left(a_{11}+2 (n-1) a_{22}+3 (n-2) (n-1) a_{33}\right) \left(a_{11}+2 n
   a_{22}+3 (n-1) n a_{33}\right).
\end{equation*}
The other two coefficients are given by cumbersome formulae which can be sent to an interested reader upon request. (In particular, the expressions for the coefficients $b_1$ consists of 53 
and for $b_2$ of 113  terms respectively).

The fourth power of commutator \eqref{eq:j4} is then a matrix with at most 13 non-vanishing diagonals (eight 
below  and four above the main diagonal). However, several of them actually vanish once the
explicit expressions for the entries are substituted. These vanishing diagonals are the fourth, the third, the second and the first above the main diagonal as well as the eighth below the main diagonal. More explicitly one has the following.   

\medskip
\noindent
The fourth diagonal above the main one is given by:
$$\mu_i-4\mu_{i+1}+6\mu_{i+2}-4\mu_{i+3}+\mu_{i+4}\equiv 0, \qquad \text{for }  i=0,\ldots .$$

\noindent
The third diagonal above the main one is given by:
\begin{equation*}
\begin{aligned}
&(-3\mu_i+6\mu_{i+1}-4\mu_{i+2}+\mu_{i+3})b_0(i)+(\mu_i+2\mu_{i+1}-4\mu_{i+2}+\mu_{i+3})b_0(i+1)+ \\
&(\mu_i-4\mu_{i+1}+2\mu_{i+2}+\mu_{i+3})b_0(i+2)+ \\
&(\mu_i-4\mu_{i+1}+6\mu_{i+2}-3\mu_{i+3})b_0(i+3)\equiv 0, \qquad \text{for }  i=0,\ldots .
\end{aligned}
\end{equation*}

\noindent
The second diagonal above the main one is given by:
\begin{equation*}
\begin{aligned}
&(-4\mu_i+7\mu_{i+1}-4\mu_{i+2}+\mu_{i+3})b_1(i+1)+(3\mu_{i+1}-4\mu_{i+2}+\mu_{i+3})b_0(i+1)^2+ \\
&(7\mu_{i+1}-8\mu_{i+2}+\mu_{i+3})b_1(i+2)+(\mu_{i+1}-2\mu_{i+2}+\mu_{i+3})b_0(i+2)^2+ \\
&(\mu_{i+1}-8\mu_{i+2}+7\mu_{i+3})b_1(i+3)+(\mu_{i+1}+2\mu_{i+2}-3\mu_{i+3})b_0(i+2)b_0(i+3)+ \\
&(\mu_{i+1}-4\mu_{i+2}+3\mu_{i+3})b_0(i+3)^2+b_0(i+1)((-3\mu_{i+1}+2\mu_{i+2}+\mu_{i+3})b_0(i+2)+ \\
&(-3\mu_{i+1}+6\mu_{i+2}-3\mu_{i+3})b_0(i+3))+ \\
&(\mu_{i+1}-4\mu_{i+2}+7\mu_{i+3}-4\mu_{i+4})b_1(i+4)\equiv 0, \qquad  \text{for }  i=0,\ldots .
\end{aligned}
\end{equation*}

\noindent
The first  diagonal above the main one is given by:
\begin{equation*}
\begin{aligned}
&-\mu _{i+2} b_0(i+2)^3+\mu _{i+3} b_0(i+2)^3+3 \mu _{i+2} b_0(i+3)b_0(i+2)^2- \\
&3 \mu _{i+3} b_0(i+3) b_0(i+2)^2-3 \mu _{i+2} b_0(i+3)^2b_0(i+2)+3 \mu _{i+3} b_0(i+3)^2 b_0(i+2)- \\
&4 \mu _{i+2} b_1(i+2) b_0(i+2)
 +2 \mu _{i+3} b_1(i+2) b_0(i+2)-4 \mu _{i+2} b_1(i+3) b_0(i+2)+ \\
&4 \mu _{i+3} b_1(i+3) b_0(i+2)-3 \mu _{i+2} b_1(i+4) b_0(i+2)+7 \mu _{i+3} b_1(i+4)
 b_0(i+2)- \\
&4 \mu _{i+4} b_1(i+4) b_0(i+2)+\mu _{i+2} b_0(i+3)^3-\mu _{i+3}
 b_0(i+3)^3-4 \mu _i b_2(i+2)- \\
&3 \mu_{i+2} b_2(i+2)+\mu _{i+3} b_2(i+2)-3
 \mu _{i+2} b_0(i+1) b_1(i+2)+ \\
&\mu _{i+3} b_0(i+1) b_1(i+2)-3 \mu _{i+2}
 b_2(i+3)-3 \mu _{i+3} b_2(i+3)+ \\
&7 \mu _{i+2} b_1(i+2) b_0(i+3)-3 \mu _{i+3}
 b_1(i+2) b_0(i+3)+4 \mu _{i+2} b_1(i+3) b_0(i+3)- \\
&4 \mu _{i+3} b_1(i+3)
 b_0(i+3)+2 \mu _{i+1} \left(3 b_2(i+2)+b_0(i+1) b_1(i+2)+\right. \\
&\left. b_1(i+2) b_0(i+2)
 +3 b_2(i+3)-2 b_1(i+2) b_0(i+3)\right)-3 \mu _{i+2} b_2(i+4)- \\
&3 \mu _{i+3}
 b_2(i+4)+6 \mu _{i+4} b_2(i+4)+2 \mu _{i+2} b_0(i+3) b_1(i+4)- \\
&4 \mu _{i+3}
 b_0(i+3) b_1(i+4)+2 \mu _{i+4} b_0(i+3) b_1(i+4)+\mu _{i+2} b_1(i+4)
 b_0(i+4)- \\
&3 \mu _{i+3} b_1(i+4) b_0(i+4)+2 \mu _{i+4} b_1(i+4) b_0(i+4)+
 \mu _{i+2} b_2(i+5)- \\
&3 \mu _{i+3} b_2(i+5)+6 \mu _{i+4} b_2(i+5)-4 \mu _{i+5}
 b_2(i+5)\equiv 0, \qquad \text{for } i=0,\ldots .
\end{aligned}
\end{equation*} 

\noindent
Finally, the eighth diagonal below the main one is given by:
\begin{equation*}
\begin{aligned}
&\left(\mu _i-4 \mu _{i+2}+6 \mu _{i+4}-4 \mu _{i+6}+\mu _{i+8}\right) b_2(i+2)
b_2(i+4) b_2(i+6) b_2(i+8)\equiv 0, \\
&\qquad \text{for }  i=0,\ldots .
\end{aligned}
\end{equation*}

The rest of diagonals are non-vanishing and their explicit expressions are rather complicated
functions of the coefficients $a_{ij}$.

\smallskip
However, if we use coefficients of polynomials that solve both the eigenvalue
problem \eqref{eq:specprob} and satisfy the  four-term recurrence relation \eqref{eq:MOP}  then the whole fourth commutator operator given by \eqref{eq:j4}
 becomes the zero matrix; this fact can be verified by plugging the coefficients from the recurrences
listed in Theorem~\ref{th:main} into $\left(\Lambda\right)_n$ and $\mathcal{L}$.

\medskip
\section {Appendix II. Computer-aided formulas relevant for  \S~\ref{sec:quadr}}
$$L_3 = (a_{32}x^2 +a_{31}x+a_{30}) \partial^3 + (x^2 +a_{21}x+a_{20}) \partial^2 +(a_{11}x+a_{10})\partial .$$

\begin{equation}\label{recgen}
xP_n(x) = P_{n+1}(x)+ \sum_{j=0}^n b_j(n)P_{n-j}(x).
\end{equation}

Putting $a_{32}=1$ we get
\begin{equation}  \label{b3}
b_3(n)=\frac{\left(8 n^9+\mathcal{O}\left(n^{8}\right)\right)}{8 \left(a_{11}+2 n-6\right) \left(a_{11}+2 n-5\right)
   \left(a_{11}+2 n-4\right) \left(a_{11}+2 n-3\right) \left(a_{11}+2
   n-2\right)}.
\end{equation}

With $a_{32}=0$ and $a_{31}=1$ we get
\begin{equation}  \label{31}
b_3(n)=\frac{\left(24 n^7+\mathcal{O}\left(n^{6}\right)\right)}{8 \left(a_{11}+2 n-6\right) \left(a_{11}+2 n-5\right)
   \left(a_{11}+2 n-4\right) \left(a_{11}+2 n-3\right) \left(a_{11}+2
   n-2\right)}.
 \end{equation}
 
 Finally,  with $a_{32}=0$, $a_{31}=0$, and $a_{30}=1$  (recall that   $a_{21}= 0$)
 we arrive at
\begin{equation}
\begin{aligned}
b_3(n)=\frac{n \left(n^2-3 n+2\right) a_{10} \left((7 n-15) a_{11}+2 a_{11}^2+6 n^2-26
   n+28\right)}{\left(a_{11}+2 n-6\right) \left(a_{11}+2 n-5\right)
   \left(a_{11}+2 n-4\right) \left(a_{11}+2 n-3\right) \left(a_{11}+2
   n-2\right)}.
\end{aligned}     \label{a20}
\end{equation}

\end{document}